%% file: main.tex
\documentclass[10pt, conference, letterpaper]{IEEEtran}
\IEEEoverridecommandlockouts
\usepackage{tikz}
\usepackage{filecontents}

\usepackage{amsmath,amsfonts,amssymb}
\usepackage{amsthm}
\usepackage{algorithmic}
\usepackage{array}
\usepackage[caption=false,font=normalsize,labelfont=sf,textfont=sf]{subfig}
\usepackage{textcomp}
\usepackage{stfloats}
\usepackage{url}
\usepackage{verbatim}
\usepackage{graphicx}
\usepackage{xcolor}
\usepackage{booktabs}
\usepackage{multirow} 
\usepackage[marginal]{footmisc}
\usepackage{balance}
\usepackage{xspace}
\usepackage{enumitem}
\usepackage[table]{xcolor}

\AtBeginDocument{%
  \providecommand\BibTeX{{%
    Bib\TeX}}}

\def\BibTeX{{\rm B\kern-.05em{\sc i\kern-.025em b}\kern-.08em
    T\kern-.1667em\lower.7ex\hbox{E}\kern-.125emX}}

\newcommand{\sysname}{RadioMaster\xspace}

\theoremstyle{plain}
\newtheorem{theorem}{Theorem}
\newtheorem{lemma}{Lemma}

\begin{document}

\title{From Intent to Air: Multi-Agent Autonomous \\ Radio Signal Generation}

\author{Jiazhen Lei, Yuxin Sha, Tianze Cao, Sihan Wang, Bingbing Wang, Zeming Yang, Fengyuan Zhu, Xiaohua Tian}

\maketitle

\begin{abstract}
Translating user intent into physical radio signals is the last critical step in wireless prototyping. It chains protocol planning, baseband synthesis, and hardware configuration. Large language models and multi-agent systems have reshaped software engineering, raising the question of whether they can solve this problem. Yet current models fail at this task, even when augmented with domain tools. Because the stages run sequentially, an error at any stage propagates downstream, so the end-to-end success rate collapses toward zero even when each stage looks locally competent. We introduce RadioMaster, a fully autonomous multi-agent framework that drives user input to verified emissions transmitted over the air. It rests on three synergistic pillars. RadioWiki grounds generation in domain knowledge to suppress hallucination. RadioAgent decomposes the fragile pipeline into independently executable and locally recoverable stages. RadioEmulator gates deployment behind closed-loop physical-layer verification. We further build RadioBench, the first benchmark for autonomous radio signal generation. Extensive real-world evaluations show that RadioMaster substantially outperforms state-of-the-art baselines in configuration viability and signal fidelity, while reducing configuration time by up to 28$\times$.
\end{abstract}


\input{1_Introduction}
\input{2_Related_Work}
\input{3_Observation}
\input{4_Design}

\input{5_Experiments}
\input{6_Conclusion}


\bibliographystyle{IEEEtran}
\balance
\bibliography{reference}

\end{document}

%% file: 1_Introduction.tex
\section{Introduction}
Radio signal generation represents the critical ``last mile'' for validating wireless communication paradigms in the physical world. Software-Defined Radio (SDR) serves as the cornerstone technology for this phase, offering unprecedented prototyping flexibility by decoupling complex signal processing from dedicated hardware~\cite{tinysdr, lei2026enablingagileambientiot}. In practice, engineers must translate high-level goals, derived from dense protocol standards and system requirements, into executable baseband algorithms and precise hardware configurations. This translation demands substantial multidisciplinary expertise~\cite{sdr1, sdr2}.

As illustrated at the top of Fig.~\ref{fig:teaser}, the traditional manual workflow for this translation is tedious and error-prone. Operators must parse thick protocol specifications, translate them into digital baseband algorithms, and then write intricate code to configure SDR platforms. These stages are tightly chained. Every stage must be correct for the final emission to succeed, yet an error at any single stage silently propagates and compounds downstream. Success is thus judged not by a plausible script or a passing simulation, but by a waveform physically emitted over the air and correctly decoded by a real receiver. This disjointed process forms a severe bottleneck. It raises the barrier to entry and slows prototyping, experimentation, and standard-compliance validation in wireless systems~\cite{wang2025bridging, cheng2025embodied}.

\begin{figure}[t]
    \centering
    \includegraphics[width=0.99\linewidth]{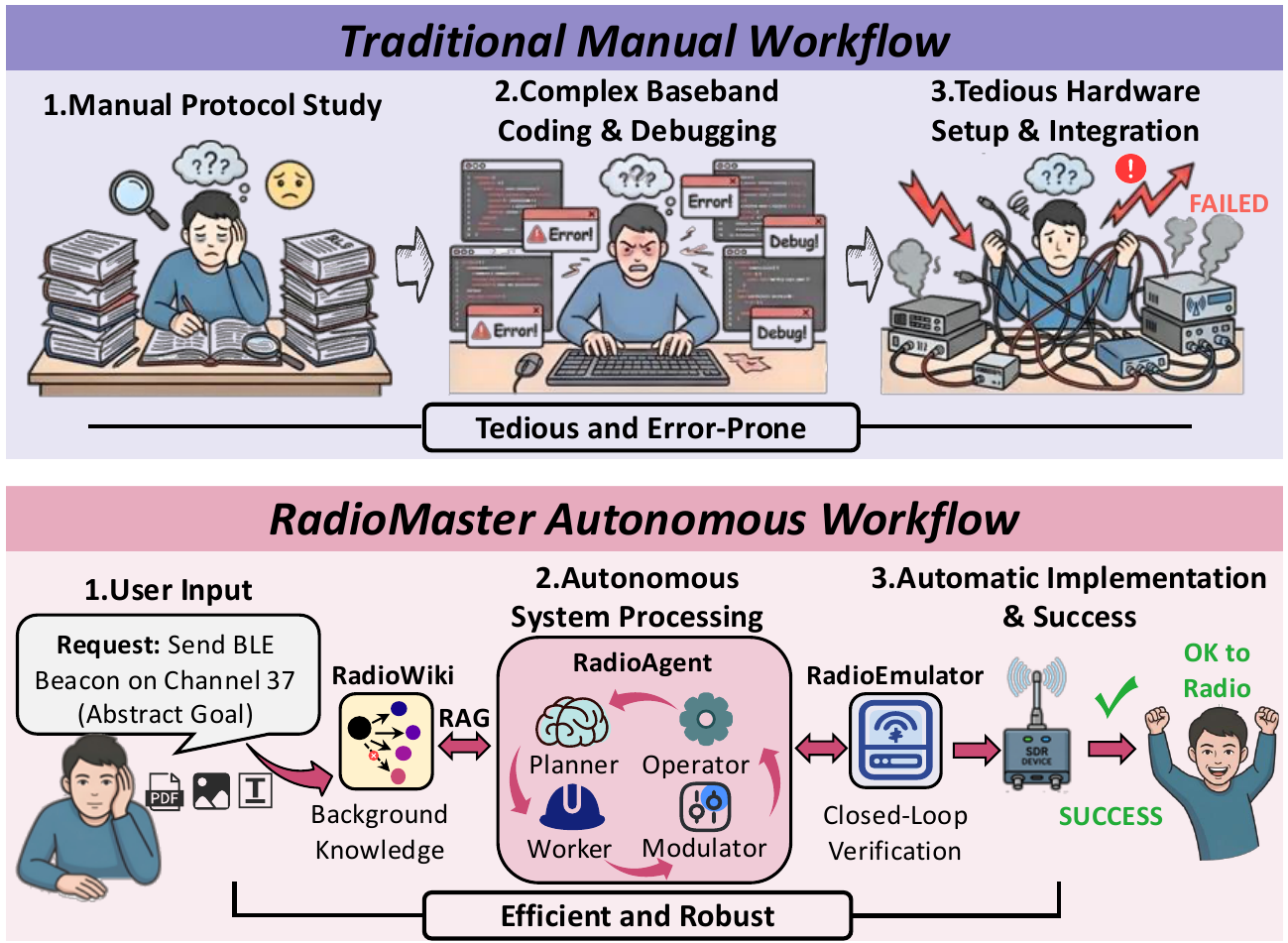}
    \caption{Radio signal generation via the traditional manual workflow versus \sysname's autonomous multi-agent framework.}
    \label{fig:teaser}
    \vspace{-5mm}
\end{figure}

Large language models (LLMs)~\cite{Claude, GPT, Qwen, DeepSeek} and multi-agent systems (MAS)~\cite{AutoIoT, IoTPilot, zhao2026agentic, ma2026autorf, luo2026emerging} have reshaped high-level software engineering and generic automation. Current frontier foundation models, such as GPT-5.5~\cite{GPT}, Claude-Opus-4.8~\cite{Claude}, Qwen3.7-Max~\cite{Qwen}, and DeepSeek-V4-Pro~\cite{DeepSeek}, alongside agentic scaffolds like Claude Code~\cite{ClaudeCode}, excel at processing general semantic logic and writing standard computational code. This agentic paradigm has even begun to penetrate specialized engineering domains. For instance, CLI agents such as Claude Code~\cite{ClaudeCode} and Codex~\cite{Codex} can now be equipped with the MATLAB Agentic Toolkit~\cite{MatlabToolkit}. This grants them programmatic access to the Communications, WLAN, and Bluetooth Toolboxes, letting them synthesize, execute, and iteratively repair signal-processing scripts against the MATLAB runtime. Such progress raises a compelling question. Can these systems resolve the difficulties inherent in end-to-end radio signal generation?

Our empirical evaluations reveal profound limitations when these generic models are applied directly to radio signal generation~\cite{xu2024penetrative, liu2025tasksense, zhao2025flexifly}. They suffer near-total failure across the complete lifecycle. Even domain-tool-augmented CLI agents, though able to iterate against a simulator, stay confined to software-level artifacts and seldom survive the crossing from simulation to over-the-air execution. The cause is structural rather than incidental. Radio signal generation is a tightly coupled chain of protocol planning, baseband synthesis, and hardware configuration. Because the stages compose sequentially, an error at any single stage propagates downstream. The end-to-end success rate is thus squeezed stage by stage and collapses towards zero, even when each stage appears locally competent. This fragility is aggravated by three failure modes that generic models exhibit along the chain: (i) \textbf{semantic misinterpretation of intricate protocols}, as models fail to maintain the tight parameter coupling and rigid formatting demanded by complex standards; (ii) \textbf{hallucination of internal APIs and functions}, where models generate syntactically plausible but nonexistent library calls; and (iii) \textbf{insensitivity to physical hardware constraints}, which yields naive configurations that fail to map digital parameters to actual transmission capabilities.

To overcome these bottlenecks, we introduce \textbf{\sysname}, a fully autonomous multi-agent framework that closes the intent-to-air execution loop and drives user input to verified radio signals emitted over the air. Since the failure is a chain that collapses stage by stage, the remedy must attack that collapse from three complementary angles. \sysname grounds every stage in verified domain knowledge, decomposes the fragile monolithic pipeline into independently executable and locally retriable stages, and gates deployment behind closed-loop physical verification. As depicted in Fig.~\ref{fig:teaser}, it instantiates these angles through three synergistic pillars. First, \textbf{RadioWiki} uses an adaptive-routing Retrieval-Augmented Generation (RAG) mechanism over a domain-specific knowledge base to ground generation and alleviate hallucination. Second, \textbf{RadioAgent} orchestrates a Planner, Worker, Modulator, and Operator that decompose the fragile pipeline into independently executable and locally recoverable stages, performing baseband processing and synthesizing precise execution scripts. Third, \textbf{RadioEmulator} runs closed-loop verification in a virtual communication environment, validating signal integrity and protocol compliance so that only reliable configurations reach the hardware. Together, these pillars carry an intent to a real over-the-air emission validated by commercial receivers, rather than stopping at code or a passive simulation.

To enable standardized quantitative evaluation for this emerging domain, we further construct \textbf{RadioBench}, the first benchmark covering autonomous radio signal generation end to end, built from expert-curated tasks under a unified and fair evaluation protocol. Unlike prior efforts that terminate at software-level code or simulator outputs, \sysname closes the intent-to-air loop and is validated on real-world hardware testbeds. Our results show that \sysname substantially outperforms existing baselines in configuration viability and physical signal fidelity.

Our contributions in this paper are as follows:
\begin{itemize}
    \item We propose \sysname, the first autonomous multi-agent framework to close the full intent-to-air execution loop, driving user intent to over-the-air emissions rather than merely generating code or simulations.
    \item We design a system architecture featuring RadioWiki for hallucination-suppressed knowledge grounding, a collaborative RadioAgent for stage-wise execution and recovery, and a simulation-based RadioEmulator for closed-loop verification, alongside RadioBench, the first expert-curated domain-specific benchmark.
    \item We conduct extensive evaluations on real-world hardware testbeds against state-of-the-art baselines, showing that \sysname generates high-fidelity, executable radio signals where generic models fail. It configures a task up to $28\times$ faster than unaided human experts and up to $7\times$ faster than Codex-assisted experts. By turning abstract intent into verified over-the-air emissions, this work lays a foundation for AI-assisted rapid prototyping and accelerated standardization in next-generation wireless systems.
\end{itemize}

\textbf{Open-source plan. } Our code, dataset and benchmark will be publicly released.

%% file: 2_Related_Work.tex
\section{Related Work}

\subsection{Radio Signal Generation}
Traditional efforts to streamline radio signal generation have yielded specialized frameworks like MATLAB toolboxes~\cite{moler2020history}, the GNU Radio ecosystem~\cite{blossom2004gnu}, and the USRP Hardware Driver (UHD)~\cite{UHD}. Recently, AI-assisted methods have been leveraged to enhance isolated stages of this pipeline, such as protocol comprehension~\cite{huang2025chat3gpp, maatouk2024tele}, baseband synthesis~\cite{thakur2024verigen, wang2024nn}, and signal recognition~\cite{zou2026rf, chen2025radiollm, tong2026wirelessagentautomatedagenticworkflow}. However, these AI-assisted systems predominantly focus on fragmented sub-tasks, largely neglecting the holistic generation lifecycle. Consequently, a profound scalability gap persists between these isolated works, preventing unified and autonomous real-world deployment.

\subsection{Multi-Agent System}
Recent research highlights a paradigm shift from monolithic LLMs to MAS for tackling complex challenges. Such systems have demonstrated remarkable efficacy in modeling autonomous interactions across domains such as LLM collaboration~\cite{guo2024large, li2024survey, yao2022react, xu2023rewoo}, embodied AI~\cite{guo2026embodied, tan2020multi, wu2025generative}, scientific problem-solving~\cite{zhao2026agentic, peng2025tree, ghafarollahi2025sciagents, fan2025ai}, automated hardware and embedded synthesis~\cite{ma2026autorf, luo2026emerging, yang2026autoembed}, and network operation~\cite{wang2026ppai, kwon2026open}. Despite these widespread successes in computational environments, extending MAS to the physical wireless domain remains significantly underexplored~\cite{xu2024penetrative}. 
Specifically, autonomously translating abstract user intents into executable physical radio emissions via multi-agent orchestration persists as a critical, unresolved challenge.

%% file: 3_Observation.tex
\section{Problem Formulation \& Motivation}

\subsection{Problem Formulation}\label{section:3.1}
A conventional intent-to-air task unfolds in three stages, namely protocol planning, baseband synthesis, and hardware configuration. We accordingly model the pipeline as a composite operator $\Phi=\Phi_{3}\circ\Phi_{2}\circ\Phi_{1}$, where $\Phi_{1}$, $\Phi_{2}$, and $\Phi_{3}$ realize the three stages. 
The pipeline yields a configuration and an emitted signal, whose outcome is judged at three increasingly strict levels with intrinsic success probabilities: an executable configuration ($\rho_{C}$), a hardware-deployable configuration ($\rho_{H}$), and an integrity-preserving emitted signal ($\rho_{S}$).

\textbf{Multiplicative Collapse.} In this pipeline, we let $p_{k}\in(0,1)$ denote the probability that stage $k$ is generated correctly.
Because a faithful emission requires all $K=3$ stages to be correct and the stages compose sequentially, the end-to-end success probability of this pipeline is
\begin{equation}
S_{\star}=\prod_{k=1}^{K}p_{k}.
\label{eq:collapse}
\end{equation}
Here $S_{\star}$ is precisely the signal-integrity rate $\rho_{S}$ attained by this monolithic pipeline, since a faithful end-to-end emission is exactly an integrity-preserving emitted signal. This multiplicative structure exposes a fundamental limitation: even when each stage is individually reliable, the product $S_{\star}$ shrinks multiplicatively across the $K=3$ stages and collapses toward zero once any single $p_{k}$ becomes small. Separately, the three success levels themselves obey a structural ordering, as formalized below.
\begin{lemma}[Metric funnel]\label{lem:funnel}
The three success probabilities satisfy $\rho_{S}\le\rho_{H}\le\rho_{C}$.
\end{lemma}
\begin{proof}
Signal integrity presupposes a successfully deployed configuration, which in turn presupposes an executable configuration; the corresponding success events are therefore nested, $\mathcal{E}_{S}\subseteq\mathcal{E}_{H}\subseteq\mathcal{E}_{C}$, and the monotonicity of probability yields the claim.
\end{proof}

In current practice, maximizing each $p_{k}$ falls to human experts, who manually study protocol specifications and proceed through trial-and-error debugging to drive every stage toward reliability. Such a workflow is slow, costly, and difficult to scale. Meanwhile, LLMs and MAS are rapidly reshaping software engineering and general-purpose automation, raising the prospect of an autonomous workflow that carries a task end to end, from intent to air, without human intervention. Yet whether LLMs and MAS can truly realize this vision remains an open question.

\begin{figure}[t]
    \centering
    \includegraphics[width=0.99\linewidth]{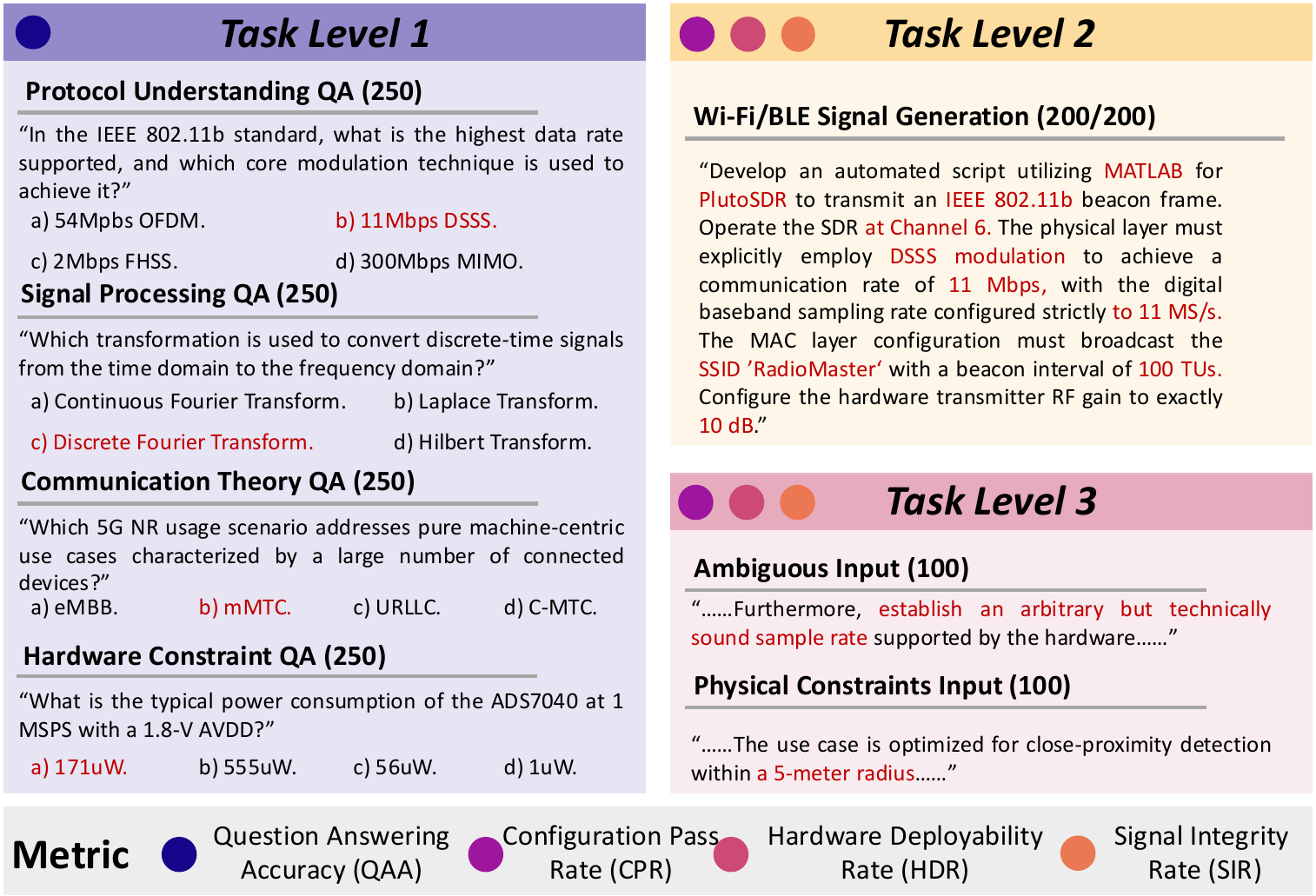}
    \caption{Overview of  RadioBench.}
    \label{fig:benchmark}
    \vspace{-5mm}
\end{figure}

\subsection{Benchmarking Autonomous Radio Signal Generation}\label{section:3.2}
To test this structural prediction, we need an evaluation standard covering the full radio signal generation lifecycle. Since no existing benchmark targets this end-to-end setting, we design \textbf{RadioBench}, the first domain-specific benchmark for quantifying the reasoning and configuration ability of LLMs and MAS in radio signal generation. It organizes evaluation into three progressive task levels.

\textit{\textbf{Task Level 1}} targets foundational domain knowledge. It contains one thousand expert-curated question-answering (QA) pairs evenly split across four core dimensions: protocol understanding (PU), signal processing (SP), communication theory (CT), and hardware constraints (HC).

\textit{\textbf{Task Level 2}} targets practical implementation, requiring models to turn explicit requirements into executable radio signals within real hardware limits. We cover the two most prevalent standards, namely Wi-Fi and Bluetooth Low Energy (BLE), and the two dominant configuration paradigms, MATLAB-based and UHD-based. This yields 400 cases, 100 per paradigm-protocol pair, denoted W-M, W-U, B-M, and B-U.

\textit{\textbf{Task Level 3}} targets deployment adaptability under complex conditions. It has 100 cases on ambiguous user requirements (AR-M and AR-U for MATLAB and UHD) and 100 cases resolving intricate demands under stringent physical constraints (PC-M and PC-U).

Both \textbf{\textit{Task Level 2}} and \textbf{\textit{Task Level 3}} require empirical validation on physical SDR platforms. We choose Wi-Fi (IEEE 802.11) and BLE as the primary protocols because both are supported by mature, widely deployed commercial receivers. This lets us verify against off-the-shelf devices whether an emitted waveform truly complies with the standard, not merely appears well-formed in a simulation. The hardware-in-the-loop validation uses packet sniffers to capture and verify over-the-air signal fidelity.

\textbf{Metrics.} We evaluate along four progressive metrics. Question Answering Accuracy (QAA) measures the factual correctness of theoretical responses. Configuration Pass Rate (CPR) measures software viability by confirming error-free script execution. Hardware Deployability Rate (HDR) measures whether configurations deploy onto real hardware without driver violations. Signal Integrity Rate (SIR) confirms whether the emitted waveform is captured and decoded by the target receiver. CPR, HDR, and SIR empirically measure the three nested success probabilities $\rho_{C}$, $\rho_{H}$, and $\rho_{S}$, and by construction inherit the funnel ordering of Lemma~\ref{lem:funnel}.

\begin{figure}[t]
    \centering
    \includegraphics[width=0.99\linewidth]{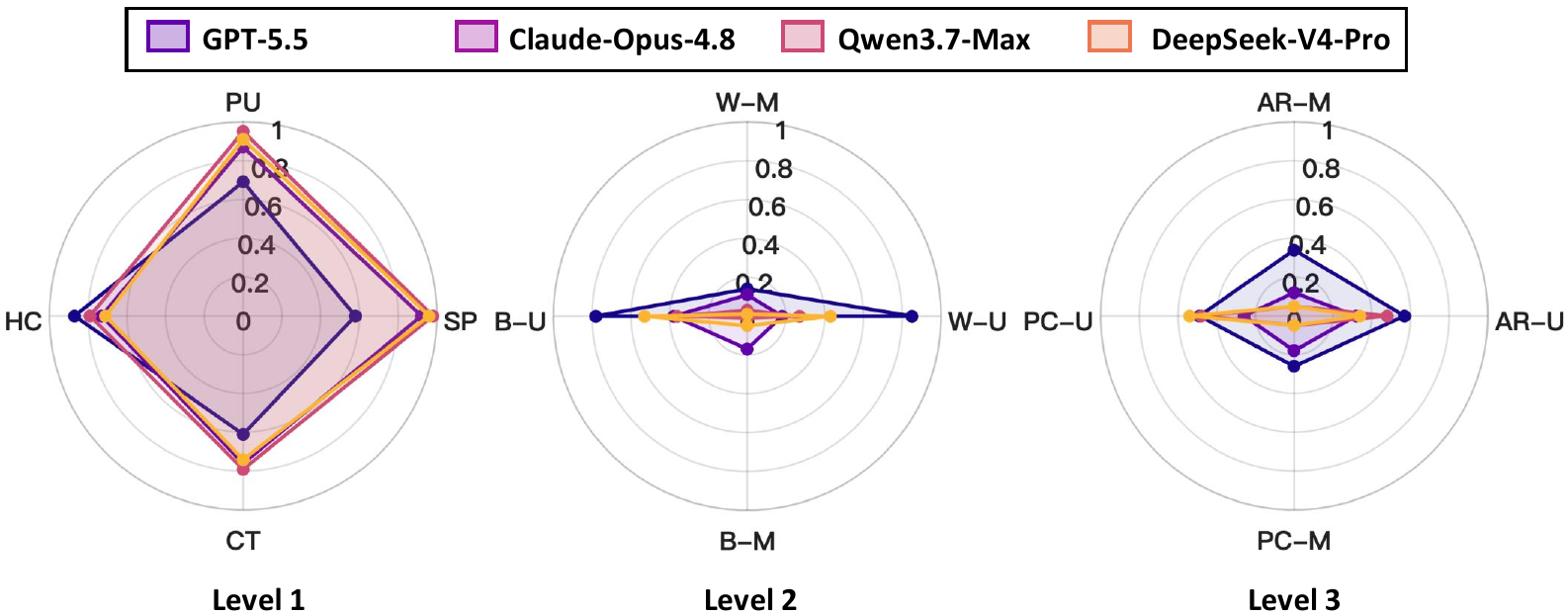}
    \caption{Preliminary evaluation of SOTA foundation models on RadioBench.}
    \label{fig:exp_prior}
    \vspace{-5mm}
\end{figure} 
\begin{figure}[t]
    \centering
    \includegraphics[width=0.85\linewidth]{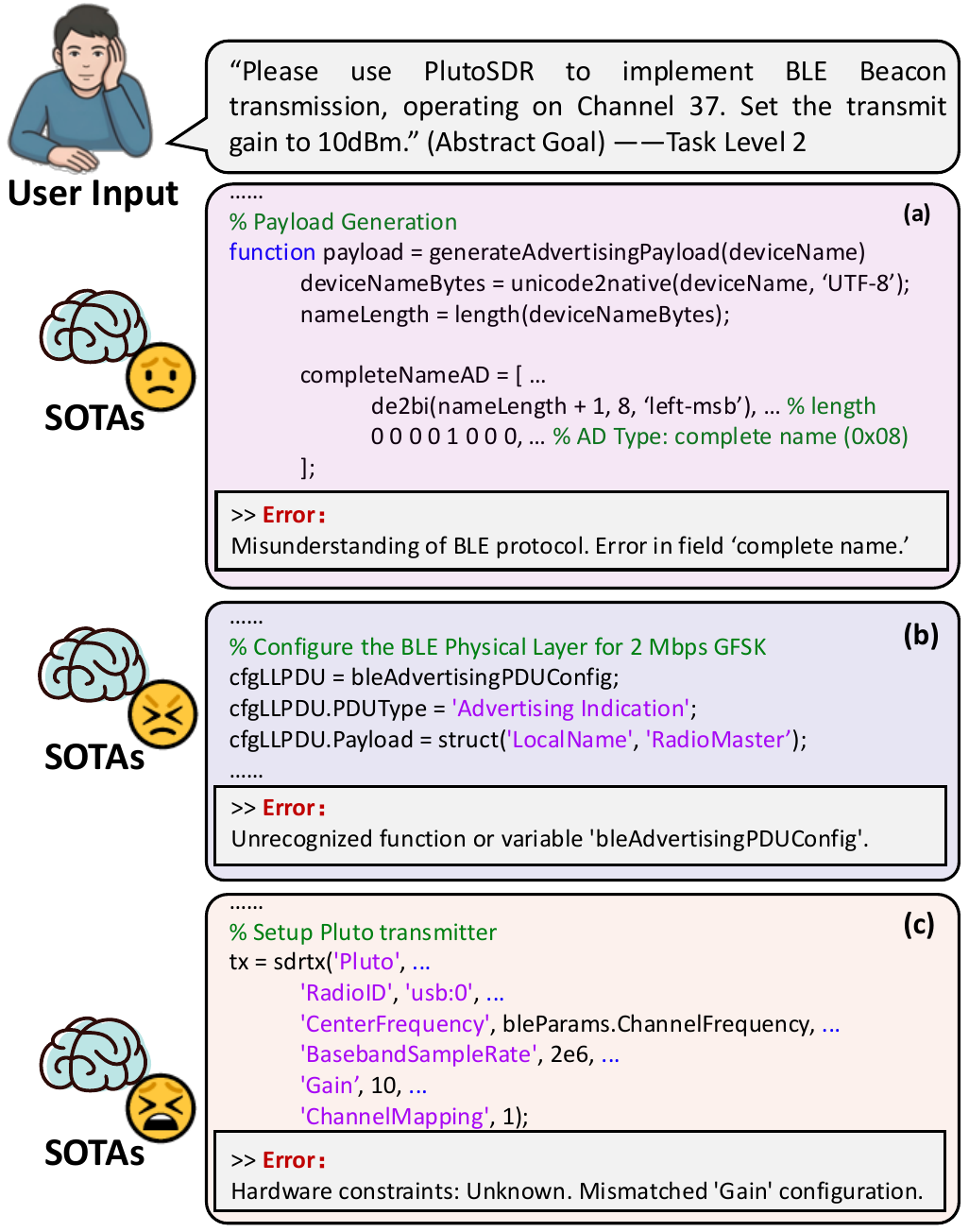}
    \caption{Analysis of fundamental limitations in current models. (a) Semantic misinterpretations. (b) Hallucinated APIs/functions. (c) Insensitivity to physical hardware constraints.}
    \label{fig:limitation}
    \vspace{-5mm}
\end{figure}

\subsection{Limitations}\label{section:3.3}
To evaluate SOTA foundation models on autonomous radio signal generation, we study four mainstream models under RadioBench, recording QAA for \textbf{\textit{Task Level 1}} and CPR for \textbf{\textit{Task Level 2}} and \textbf{\textit{Task Level 3}}. As shown in Fig.~\ref{fig:exp_prior}, models handle basic knowledge retrieval well, yet degrade sharply as tasks advance along the pipeline. Inspecting the generated scripts reveals three fundamental limitations.

\textbf{L1. Semantic misinterpretation of intricate protocols.} Lacking specialized wireless data, models miss the tight parameter coupling and rigid formatting required by standards, producing flawed payloads that fail validation (Fig.~\ref{fig:limitation}(a)). In Eq.~\eqref{eq:collapse}, this depresses $p_{1}$. Since a misconstrued field surfaces only as a corrupted waveform, it caps $\rho_{S}$ and drives $S_{\star}$ toward zero.

\textbf{L2. Hallucination of internal APIs and functions.} Without exact documentation, models synthesize plausible but nonexistent commands. In Fig.~\ref{fig:limitation}(b), a BLE script invokes the non-existent \texttt{bleAdvertisingPDUConfig} instead of the correct \texttt{bleLLAdvertisingChannelPDUConfig}, breaking the pipeline. This suppresses $p_{2}$; as an executable configuration precedes every downstream stage, it caps $\rho_{C}$ and propagates to $S_{\star}$.

\textbf{L3. Insensitivity to physical hardware constraints.} Models often ignore the target platform's operational limits. For example, a raw integer for a $10\,\mathrm{dBm}$ gain in MATLAB fails, since the numeric parameter is not linear in absolute power (Fig.~\ref{fig:limitation}(c)). This erodes $p_{3}$, so even a clean script fails to deploy, depressing $\rho_{H}$ and hence $S_{\star}$.

\textit{\textbf{Motivation.}} The bottlenecks thus stem from protocol misunderstanding, hallucinated functions, and unawareness of physical constraints. Efficiently integrating external domain knowledge with internal reasoning for robust autonomous configuration remains an open challenge.

%% file: 4_Design.tex
%
\begin{figure*}[t]
    \centering
    \includegraphics[width=0.85\linewidth]{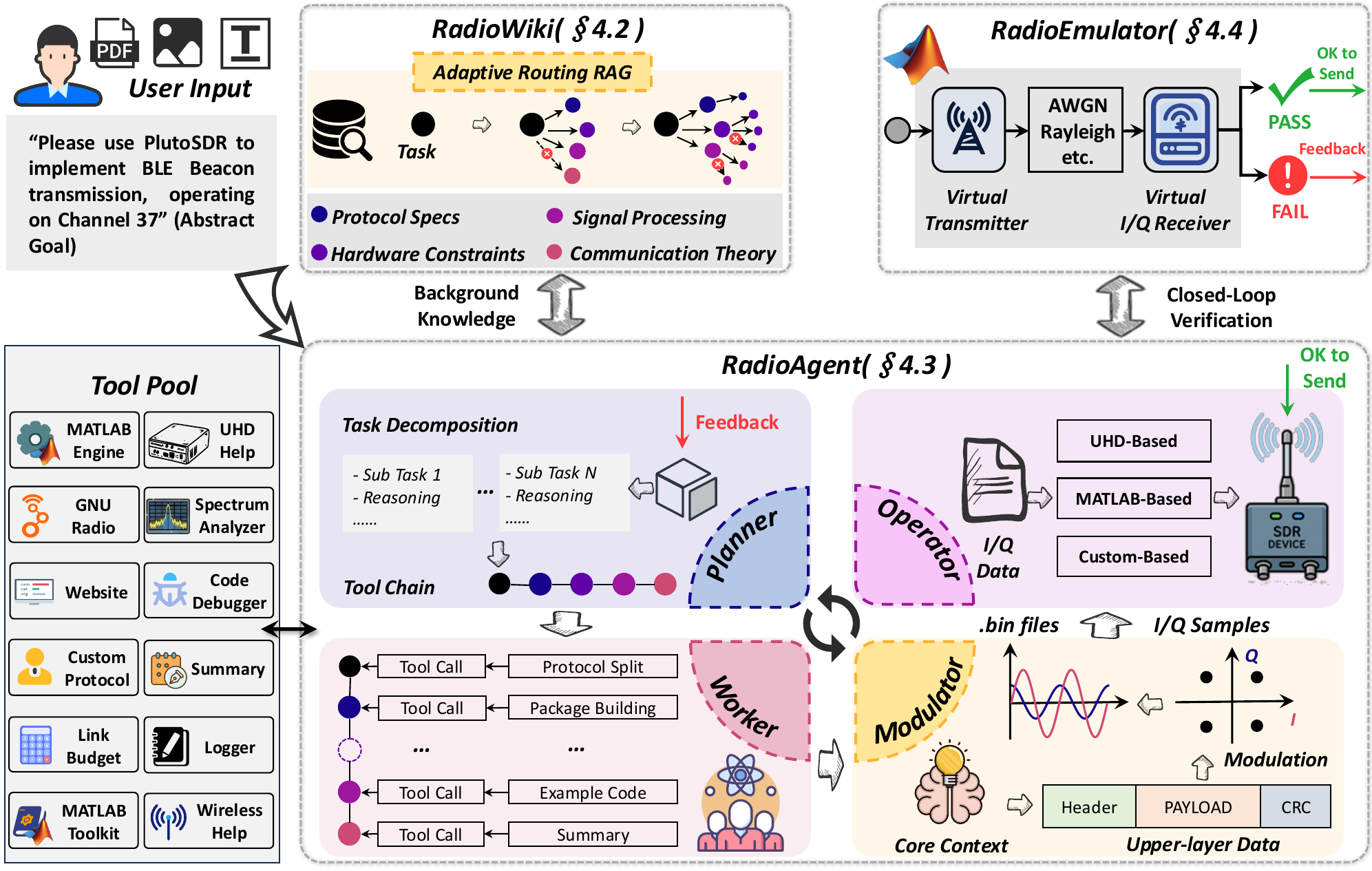}
    \caption{\sysname System Block Diagram.}
    \label{fig:overview}
    \vspace{-5mm}
\end{figure*}

\section{Design}
\subsection{Overview of \sysname}
As shown in Fig.~\ref{fig:overview}, we introduce \textbf{\sysname}, a fully autonomous multi-agent system for the full radio signal generation lifecycle, whose architecture directly mitigates the multiplicative collapse. It rests on three synergistic pillars, each formalized as a design principle and realized by a dedicated module. Under \textbf{P1 (Background Knowledge Grounding)}, RadioWiki injects verified external evidence via adaptive-routing RAG over a domain knowledge base, recovering erroneous attempts and lifting $p_{k}$. Under \textbf{P2 (Multi-Agent Collaboration)}, RadioAgent decomposes the monolithic operator $\Phi$ into independently executed, locally retriable stages, protocol planning, baseband synthesis, and hardware configuration, turning the fragile product into a term-wise larger one. Under \textbf{P3 (Closed-loop Verification Gating)}, RadioEmulator inserts a verification gate $\Phi_{v}$, absent from monolithic baselines, that validates the configuration by simulation and filters physically defective waveforms, the residual regime of high CPR yet collapsed SIR, routing anomalies back to RadioAgent for refinement. On success, RadioAgent's Operator synthesizes a pipeline configuration file that runs the full lifecycle via one command.

\subsection{Background Knowledge Grounding}\label{section:4.2}
Current LLMs rely on generalized corpora, lacking deep domain expertise and recent wireless protocol advances. To bridge this gap, we introduce \textbf{RadioWiki}, an extensible knowledge base tailored for the radio-frequency domain (Fig.~\ref{fig:radiowiki}). We pair it with a specialized RAG system that autonomously retrieves key information, integrates real-time online data, and compiles findings into a structured format, providing fact-grounded context that augments the reasoning and execution of downstream agentic modules.

\subsubsection{Data Collection}
These span ten prevailing wireless protocols: Wi-Fi, BLE, 3GPP 4G LTE, 3GPP 5G NR, LoRaWAN, Zigbee, SparkLink, RFID, UWB, and Ambient IoT. They further include mainstream SDR platforms such as the USRP series and ADALM-Pluto SDR, RF front-end datasheets, communication theory literature, and software toolbox guidelines.

We partition this knowledge base so that adding new materials only reconstructs the affected sections. Embedding models then transform diverse texts and diagrams into high-dimensional vectors for storage. The resulting adaptable database provides the granular context needed for robust end-to-end radio signal generation and remains easily upgradable as standards evolve and hardware iterates.

\subsubsection{Adaptive Routing RAG}
To support scalable and accurate knowledge grounding, we instantiate an adaptive routing RAG. Given a query $q$, this RAG selects a small set of relevant repositories, retrieves high-value evidence, and supplies structured context for downstream generation:
\begin{equation}
\hat{y}\sim p_{\theta}\!\left(y \mid q,\mathcal{R}(q)\right),
\end{equation}
where $\mathcal{R}(q)$ is the retrieved context set.
Let $\mathcal{C}=\{1,\dots,N\}$ denote domain repositories (e.g., wireless protocols, baseband toolchains, SDR hardware manuals). Adaptive routing follows a two-stage design to balance efficiency and robustness.

\textbf{Stage 1: Lexical fast routing.}
We compute a domain-weighted lexical confidence:
\begin{equation}
s_{\text{lex}}(j\mid q)=\sum_{t\in\phi(q)} w_t\,\mathbf{1}[t\in\mathcal{K}_j],
\end{equation}
where $\phi(q)$ is normalized tokenization, $\mathcal{K}_j$ is the keyword inventory of repository $j$, and $w_t$ upweights high-precision technical entities (e.g., chipset IDs, standard clauses, toolbox names). If
\begin{equation}
\max_{j} s_{\text{lex}}(j\mid q)\ge \tau_{\text{lex}},
\end{equation}
routing is finalized directly, avoiding embedding-based routing.

\textbf{Stage 2: Semantic centroid routing.}
If the lexical confidence $\max_{j} s_{\text{lex}}(j\mid q)$ falls below $\tau_{\text{lex}}$, we compute the query embedding $\mathbf{e}_q$ and repository centroids $\mathbf{c}_j$, and score:
\begin{equation}
s_{\text{sem}}(j\mid q)=\cos(\mathbf{e}_q,\mathbf{c}_j).
\end{equation}
A single-repository decision is accepted only under confidence and margin constraints:
\begin{equation}
s_{(1)}\ge \tau_{\text{abs}}, \qquad s_{(1)}-s_{(2)}\ge \tau_{\text{margin}},
\end{equation}
where $s_{(1)}$ and $s_{(2)}$ are the top-2 semantic scores. Otherwise, we use top-$M$ multi-repository fallback. Here, the two hyperparameters $\tau_{\text{abs}}$ and $\tau_{\text{margin}}$ control reliability and uniqueness, respectively.

\textbf{Hybrid retrieval and fusion.}
For the routed repository set, we acquire evidence through two channels, querying each repository with a dense retriever over the vector index $\mathcal{V}_j$ and a sparse retriever over the BM25 index $\mathcal{B}_j$, coupling semantic coverage with exact lexical matching. Let $\pi_v(d\mid q)$ and $\pi_b(d\mid q)$ be the normalized relevance scores of the dense ($\mathrm{DC}$) and sparse ($\mathrm{SC}$) channels for candidate chunk $d$. We aggregate them as
\begin{equation}
S(d\mid q)=\alpha\,\pi_v(d\mid q)+(1-\alpha)\,\pi_b(d\mid q)+\beta\,\mathbf{1}[d\in\mathrm{DC}\cap\mathrm{SC}],
\end{equation}
where $\alpha\in[0,1]$ controls the dense-sparse trade-off and the agreement bonus $\beta$ favors candidates supported by both views, which empirically stabilizes ranking under terminology drift and query paraphrase. The first term emphasizes semantic proximity, the second preserves high-precision entity-level matching such as standard clauses, chipset identifiers, and toolbox names, and the third acts as a consistency prior promoting cross-view corroboration. After fusion, we suppress near-duplicates to remove redundant chunks, then rank the deduplicated pool globally to yield a compact, high-fidelity context set for downstream generation.
\begin{figure}[t]
    \centering
    \includegraphics[width=0.9\linewidth]{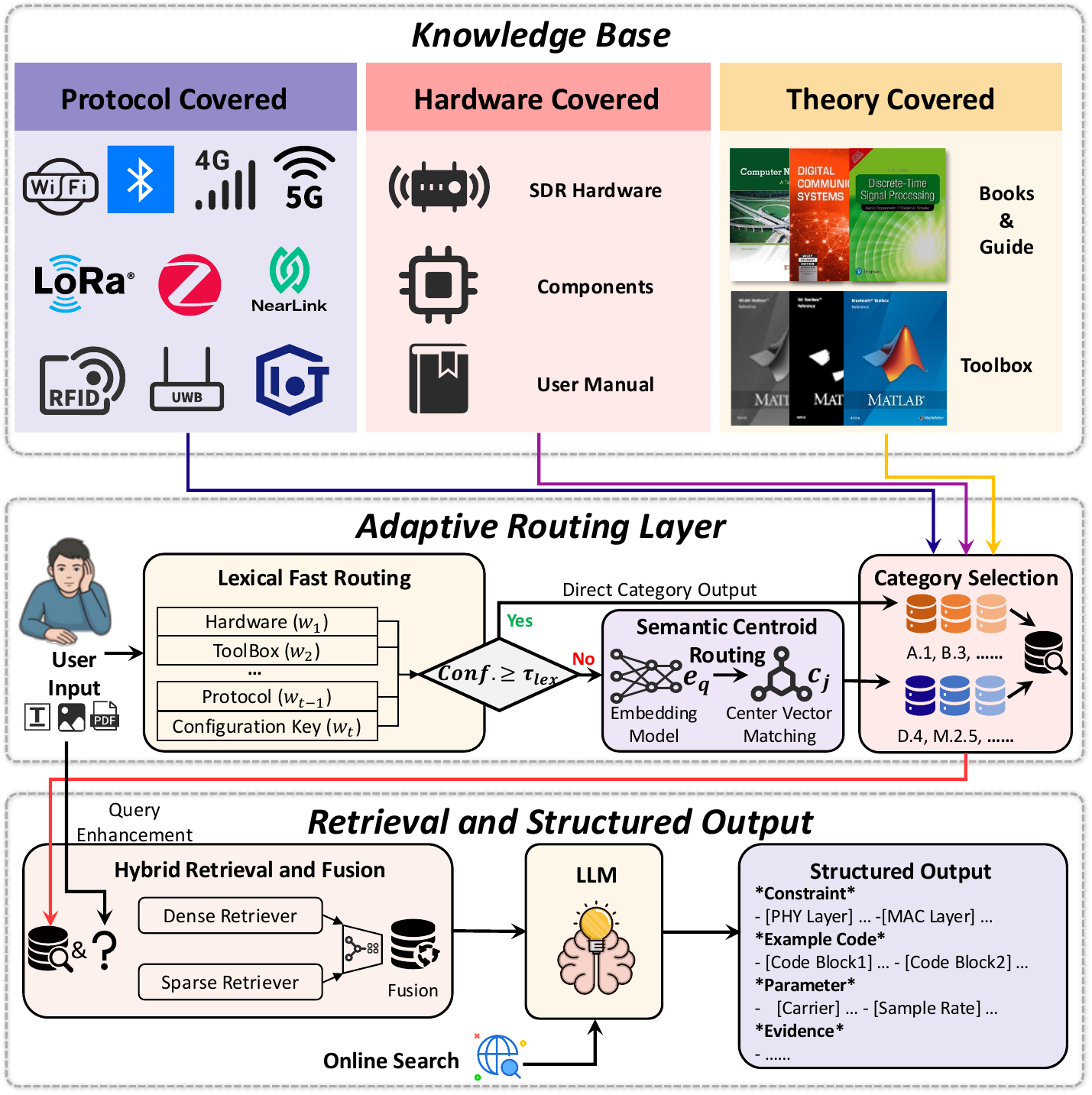}
    \caption{Architecture of RadioWiki.}
    \label{fig:radiowiki}
    \vspace{-5mm}
\end{figure}

\subsubsection{Structured Generation}
An LLM then synthesizes the heterogeneous data into a standardized format that guides RadioAgent's reasoning. It aggregates candidate contexts from the local hybrid retriever and dynamic online search. Guided by engineered prompts, it then filters out noise and reorganizes the remaining information into a cohesive, structured output. As shown in Fig.~\ref{fig:radiowiki}, this distills essential protocol constraints, hardware-specific parameters, authoritative evidence, and verified code snippets into a unified paradigm.
RadioWiki thus turns fragmented domain knowledge into deterministic context, mitigating hallucination and equipping RadioAgent with precise specifications for robust SDR configuration. This realizes principle~\textbf{P1 (Background Knowledge Grounding)} and directly addresses the semantic misinterpretation of intricate protocols (\textbf{L1}), while partially alleviating the hallucination of internal APIs and functions (\textbf{L2}), leaving residual generation-time errors to the downstream verification loop.

\subsection{Multi-Agent Collaboration}\label{section:4.3}
To orchestrate radio signal generation, we propose \textbf{RadioAgent}, an autonomous multi-agent framework with four synergistic roles, the \textit{Planner}, \textit{Worker}, \textit{Modulator}, and \textit{Operator}. They operate within a tightly coordinated pipeline, dividing responsibilities to form a robust iterative generation and refinement loop. We further provide a comprehensive tool pool for dynamic tool invocation by each agent, enhancing their capabilities across diverse tasks.

\subsubsection{Planner}
As shown in Fig.~\ref{fig:overview}, the Planner is the cognitive reasoning center. Given the user input and the structured context from RadioWiki, it decomposes the overall task into a chronological sequence of atomic subtasks with explicit logical evidence, then synthesizes a deterministic tool chain that dictates the execution workflow for downstream modules. The Planner also adapts by integrating diagnostic feedback from RadioEmulator's closed-loop verification into its contextual memory, enabling replanning and self-correction in later iterations.

\subsubsection{Worker}
The Worker executes the tool chain from the Planner, dynamically invoking specialized utilities from the tool pool for operations such as protocol splitting, package building, and code debugging. For deterministic execution, every invocation follows a rigid template within \texttt{<Call></Call>} tags. Three semantic delimiters define its parameters: \texttt{<Tool></Tool>} names the utility, \texttt{<Query></Query>} holds the input payload, and \texttt{<Purpose></Purpose>} states the objective. The sequence ends with an \texttt{<EndCall></EndCall>} tag, and the synthesized results form the summary that guides the next stage.

\subsubsection{Modulator}
Using the Worker's structured summary, the Modulator runs the core digital baseband pipeline. It assembles upper-layer protocol data packets and generates the corresponding bitstream. Guided by user requirements and physical-layer constraints, it applies digital modulation to map the bitstream into the complex domain. The resulting In-phase and Quadrature (I/Q) samples are saved to a \texttt{.bin} file, providing the exact baseband waveform for physical transmission.

\subsubsection{Operator}
The Operator bridges digital baseband synthesis and physical deployment by generating executable configuration scripts for the target platform, including UHD APIs, MATLAB environments, and custom protocols. It ingests the Modulator's \texttt{.bin} file and treats the samples as the sequence destined for aerial transmission. For reliability, a strict gate initiates deployment only after RadioEmulator passes its validation loop. Once authorized, the Operator synthesizes the final \texttt{run\_pipeline} executable, letting users run the fully verified generation lifecycle via a single command.

Collectively, this division of labor instantiates principle~\textbf{P2 (Multi-Agent Collaboration)}. The four roles realize the three stages of the operator $\Phi$: the Planner and Worker carry out protocol planning ($\Phi_{1}$), the Modulator performs baseband synthesis ($\Phi_{2}$), and the Operator handles hardware configuration ($\Phi_{3}$), so the pipeline reduces to $K=3$ independently executed stages. By the grounding of Section~\ref{section:4.2} (P1), each per-stage correctness is lifted from $p_{k}\in(0,1)$ to $\tilde{p}_{k}=p_{k}+(1-p_{k})\,g_{k}$ via a recovery fraction $g_{k}\in[0,1]$, so $\tilde{p}_{k}\ge p_{k}$. Within stage $k$, an erroneous output is detected with probability $\delta_{k}\in[0,1]$ and locally retried up to $r_{k}$ times. Writing $u_{k}=(1-\tilde{p}_{k})\,\delta_{k}$ for a detected-and-retried failure, the effective stage correctness after retries is
\begin{equation}
q_{k}=\tilde{p}_{k}\,\frac{1-u_{k}^{\,r_{k}}}{1-u_{k}}.
\label{eq:qk}
\end{equation}
Replacing the monolithic operator with $K$ independently grounded and retriable stages then lifts the end-to-end success rate.

\begin{theorem}[Decomposition and grounding gain]\label{thm:chain}
Under independent per-stage execution, the end-to-end success probability obeys
\begin{equation}
\rho_{S}\;\ge\;\prod_{k=1}^{K}q_{k}\;\ge\;\prod_{k=1}^{K}\tilde{p}_{k}\;\ge\;\prod_{k=1}^{K}p_{k}\;=\;S_{\star},
\label{eq:chain}
\end{equation}
where $S_{\star}$ is the monolithic single-agent baseline of Eq.~\eqref{eq:collapse}.
\end{theorem}
\begin{proof}
Correct outputs at all stages are sufficient for end-to-end success and the stages execute independently, giving the first inequality. Since $0\le u_{k}<1$, the retry factor in \eqref{eq:qk} expands as $\tfrac{1-u_{k}^{\,r_{k}}}{1-u_{k}}=\sum_{i=0}^{r_{k}-1}u_{k}^{\,i}\ge1$, hence $q_{k}\ge\tilde{p}_{k}$; and $g_{k}\ge0$ with $p_{k}<1$ gives $\tilde{p}_{k}\ge p_{k}$. Multiplying the three chains over $k$ yields \eqref{eq:chain}, with strict inequality whenever some stage has $g_{k}>0$, or $\delta_{k}>0$, $r_{k}>1$, and $\tilde{p}_{k}<1$.
\end{proof}

Equation~\eqref{eq:chain} shows how RadioWiki (via $g_{k}$) and the multi-agent decomposition (via $\delta_{k}$ and $r_{k}$) counteract the multiplicative collapse $S_{\star}=\prod_{k}p_{k}$ term by term, converting the bare product into $\prod_{k}q_{k}$. The within-stage detect-and-retry ($\delta_{k}$, $r_{k}$) further closes API hallucination (\textbf{L2}): a surviving hallucinated call breaks execution, is detected in place, and is repaired before propagating to $S_{\star}$.

\subsection{Closed-Loop Verification Gating}\label{section:4.4}
Unlike software agents that debug syntax through virtual compilation, physical radio deployment is far more complex. In \sysname, passing baseband synthesis without software exceptions does not guarantee transmission fidelity, since emitted waveforms may still carry structural defects such as inaccurate transmission rates or malformed protocol fields.

To address this, we design \textbf{RadioEmulator}, a closed-loop verification safeguard preceding hardware emission. It builds an end-to-end virtual communication system in MATLAB, with a virtual transmitter, realistic channel models, and a virtual receiver, to evaluate the synthesized I/Q samples. On anomalies, it routes diagnostic feedback to the Planner for iterative refinement; on success, it authorizes the Operator to synthesize the final \texttt{run\_pipeline} executable. This ensures that only verified, high-fidelity configurations reach the physical platform.

This realizes principle~\textbf{P3 (Closed-loop Verification Gating)}. The verification operator $\Phi_{v}$ passes each waveform through the virtual receiver and summarizes its fidelity by the receiver-side metric vector $\hat{m}=(\mathrm{EVM},\mathrm{BER},\mathrm{PAPR})$, measured after channel propagation and demodulation. Deployment is authorized only when these metrics fall within the acceptance region
\begin{equation}
\mathcal{A}=\big\{\hat{m}:\ \mathrm{EVM}<\tau_{\mathrm{EVM}},\ \mathrm{BER}<\tau_{\mathrm{BER}},\ \mathrm{PAPR}<\tau_{P}\big\},
\label{eq:accept}
\end{equation}
with protocol-dependent thresholds. Let $\omega$ be the fraction of upstream waveforms that are physically faithful and $\eta$ the probability that the gate rejects a defective one.

\begin{theorem}[Deployment precision under gating]\label{thm:gate}
Assuming faithful waveforms are always accepted, the precision of the deployed configurations is
\begin{equation}
\Pi=\frac{\omega}{\omega+(1-\omega)(1-\eta)},
\label{eq:gate}
\end{equation}
which increases monotonically in the gate sensitivity $\eta$.
\end{theorem}
\begin{proof}
A faithful waveform passes with probability one and a defective one with probability $1-\eta$, so the accepted waveforms number $\omega+(1-\omega)(1-\eta)$, of which the faithful fraction is \eqref{eq:gate}.
\end{proof}

Hence, a stricter gate ($\eta\to1$) drives $\Pi\to1$, whereas removing it ($\eta\to0$) leaves residual defects exposed. Since the gate cannot fabricate faithful waveforms, this higher $\Pi$ raises the signal-integrity rate $\rho_{S}$ (SIR) only insofar as rejected outputs are repaired. The upstream rates $\rho_{C}$ and $\rho_{H}$ are fixed before the gate and stay invariant to $\eta$, so the gate reshapes only $\rho_{S}$. The same diagnostics drive a bounded refinement loop where each round repairs a residual defect with probability at least $\kappa$, so the unresolved-defect probability decays geometrically as $(1-\kappa)^{\mathcal{G}}$ over $\mathcal{G}$ rounds. Thus, the gate resolves the two failures no text-level check can reach: hardware insensitivity (\textbf{L3}), which surfaces only after channel propagation and depresses $\rho_{H}$; and residual protocol misinterpretation (\textbf{L1}), which manifests as a corrupted waveform collapsing $\rho_{S}$ under high CPR.

%% file: 5_Experiments.tex
\section{Evaluation}
\begin{figure}[t]
    \centering
    \includegraphics[width=0.99\linewidth]{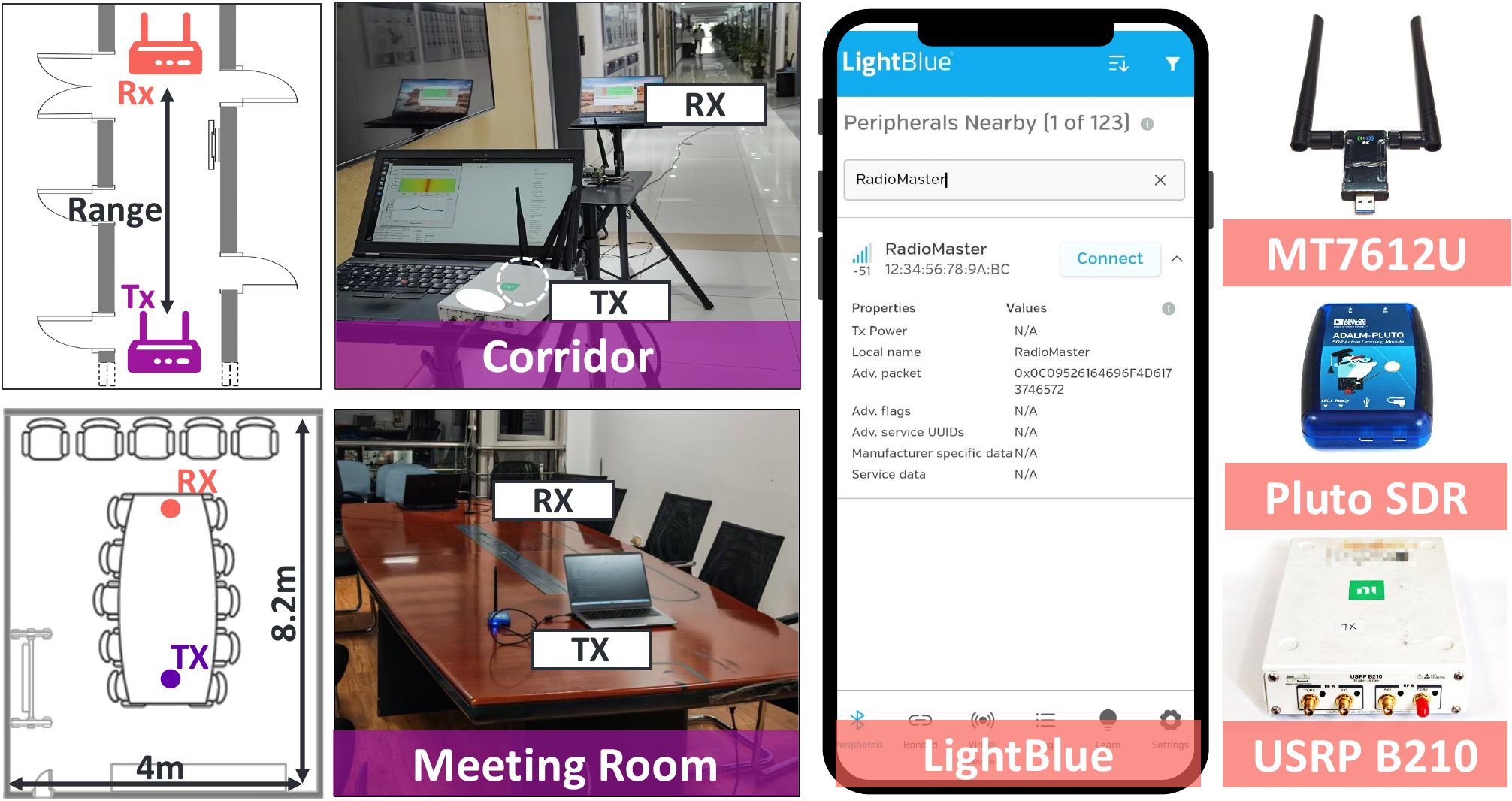}
    \caption{Experimental Setup.}
    \label{fig:setup}
    \vspace{-3mm}
\end{figure}
\begin{table}[t]
\centering
\caption{Backbone selection for \sysname on RadioBench.}
\label{table: result_different_infra}
\resizebox{0.49\textwidth}{!}{
\begin{tabular}{ l | c | cc | cc | cc }
\toprule
\multirow{2}{*}{\textbf{Method}} & \multicolumn{1}{c}{\textbf{QAA}} & \multicolumn{2}{|c|}{\textbf{CPR}} & \multicolumn{2}{c|}{\textbf{HDR}} & \multicolumn{2}{c}{\textbf{SIR}} \\
& \multicolumn{1}{c}{Level 1} & \multicolumn{1}{|c}{Level 2} & \multicolumn{1}{c|}{Level 3} & \multicolumn{1}{c}{Level 2} & \multicolumn{1}{c|}{Level 3} & \multicolumn{1}{c}{Level 2} & \multicolumn{1}{c}{Level 3} \\
\midrule
Qwen3.7-Max~\cite{Qwen} & 0.97 & 0.88 & 0.79 & 0.85 & 0.77 & 0.79 & 0.67 \\
DeepSeek-V4-Pro~\cite{DeepSeek} & 0.95 & 0.86 & 0.77 & 0.84 & 0.75 & 0.77 & 0.66 \\
Claude-Opus-4.8~\cite{Claude} & 0.97 & 0.89 & 0.79 & 0.87 & 0.78 & 0.79 & 0.68 \\
GPT-5.5~\cite{GPT} & 0.97 & 0.88 & 0.80 & 0.87 & 0.77 & 0.80 & 0.68 \\
\bottomrule
\end{tabular}
}
\vspace{-5mm}
\end{table}

\begin{figure*}[t]
    \centering
    \includegraphics[width=0.9\linewidth]{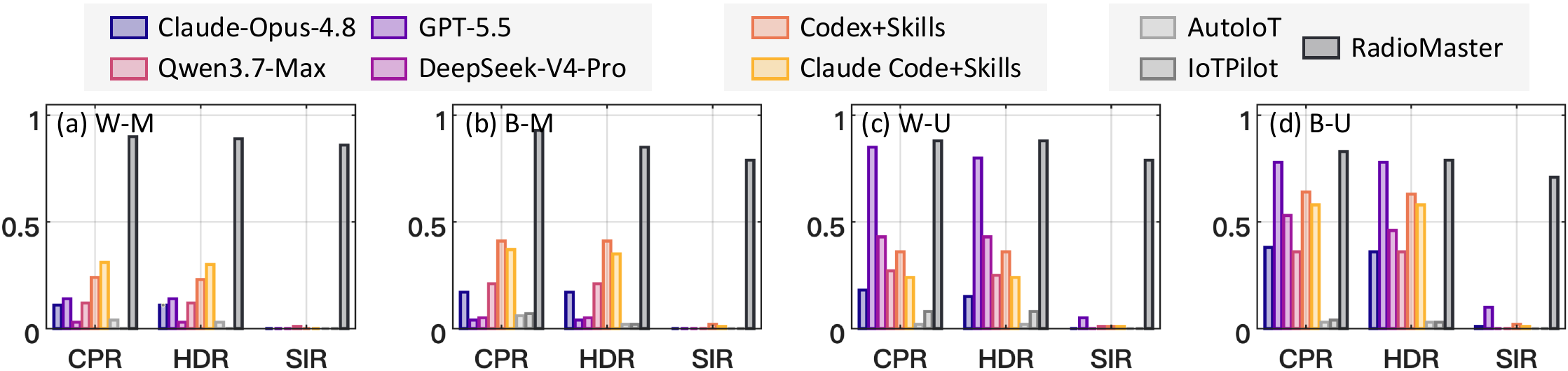}
    \caption{Main results on \textbf{\textit{Task Level 2}}. (a) W-M and (b) B-M correspond to MATLAB-based Wi-Fi and BLE; (c) W-U and (d) B-U correspond to UHD-based Wi-Fi and BLE.}
    \label{fig:mainresult_level2}
    \vspace{-3mm}
\end{figure*} 
\begin{figure*}[t]
    \centering
    \includegraphics[width=0.9\linewidth]{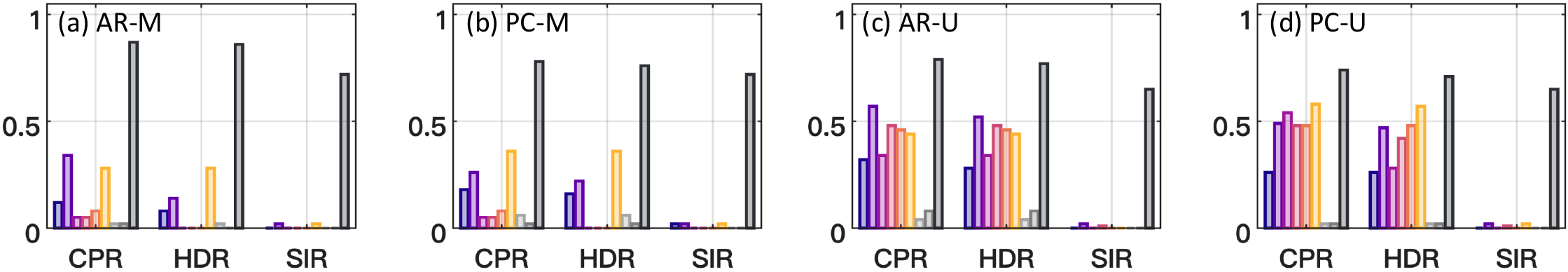}
    \caption{Main results on \textbf{\textit{Task Level 3}}. (a) AR-M and (b) PC-M correspond to MATLAB-based ambiguous requirements and physical constraints; (c) AR-U and (d) PC-U correspond to UHD-based ambiguous requirements and physical constraints.}
    \label{fig:mainresult_level3}
    \vspace{-5mm}
\end{figure*} 

\subsection{Experimental Settings}\label{section:5.1}
\subsubsection{Baselines}
Building on the benchmark design in Section~\ref{section:3.2}, we evaluate our \sysname framework against a diverse set of state-of-the-art methods under RadioBench, grouped as follows.

\textbf{Foundation Models.} We compare against several SOTA LLMs, including Qwen3.7-Max~\cite{Qwen}, DeepSeek-V4-Pro~\cite{DeepSeek}, Claude-Opus-4.8~\cite{Claude}, and GPT-5.5~\cite{GPT}. These models represent the frontier of AI, with strong capabilities in complex task planning, semantic comprehension of technical protocols, and context-aware code generation.

\textbf{Open-Source Multi-Agent Systems.} MAS for the broader Internet of Things (IoT) ecosystem has proliferated rapidly, and reliable radio signal generation is a foundational pillar of their network connectivity. We evaluate representative frameworks, including AutoIoT~\cite{AutoIoT} and IoTPilot~\cite{IoTPilot}, which parse ambiguous user requirements, orchestrate task workflows, and synthesize configuration files for generic IoT deployments.

\textbf{CLI Agent with MATLAB Agentic Toolkit.} To test whether general-purpose agentic coding assistants with domain tool access can bridge the intent-to-air gap, we equip two state-of-the-art coding CLIs, Claude Code~\cite{ClaudeCode} and Codex~\cite{Codex}, with the MATLAB Agentic Toolkit~\cite{MatlabToolkit}. The toolkit grants programmatic access to the Communications, WLAN, and Bluetooth toolboxes, letting the agent synthesize, run, and iteratively repair configuration scripts directly against the MATLAB runtime.

\subsubsection{Metrics}
All evaluations are conducted within RadioBench, using a tiered metric suite matched to each task level as detailed in Section~\ref{section:3.2}. i) For \textbf{\textit{Task Level 1}}, dominated by foundational knowledge queries, we use Question Answering Accuracy (QAA) to measure the precision of information retrieval and domain reasoning. ii) For the practical challenges in \textbf{\textit{Task Level 2}} and \textbf{\textit{Task Level 3}}, we adopt a progressive three-tier system assessing the depth and physical viability of task completion, namely Configuration Pass Rate (CPR), Hardware Deployability Rate (HDR), and Signal Integrity Rate (SIR).

\subsection{Implementation}\label{section:5.2}
\subsubsection{Details}
To assess successful over-the-air emission, we deploy generated configurations on physical hardware testbeds. As shown in Fig.~\ref{fig:setup}, the PlutoSDR and USRP B210 act as transmitters for the MATLAB-based and UHD-based methods, respectively. For validation, we use commercial receivers. A MediaTek MT7612U interface with Omnipeek decodes Wi-Fi packets, while the LightBlue application uses a smartphone's BLE receiver to parse BLE advertising packets. This end-to-end verification ensures full interoperability with standard commercial devices.

RadioBench test cases are evenly split between the two configuration paradigms. We build the dataset by LLM-assisted generation followed by manual sanitization and expert calibration to guarantee strict protocol compliance.

\subsubsection{Backbone model}
To choose a backbone for \sysname, we instantiate the framework with several representative LLMs and evaluate end-to-end performance (Table~\ref{table: result_different_infra}). Results are largely consistent across backbones, indicating that \sysname's effectiveness stems from its architecture rather than any particular model. We thus adopt Qwen3.7-Max~\cite{Qwen} as the default backbone. For a fair comparison, both the AutoIoT~\cite{AutoIoT} and IoTPilot~\cite{IoTPilot} systems and the two CLI agents run on the same Qwen3.7-Max backbone.

Finally, to bound autonomous code-generation cost, RadioAgent limits the code debugger to 10 local iterations. If closed-loop validation fails, RadioEmulator returns diagnostic feedback for regeneration, capped at 5 global attempts to balance signal fidelity and computational overhead.

\begin{table}[t]
\centering
\caption{Quantitative Comparisons on RadioBench in \textbf{\textit{Task Level 1}}. The best results are marked in bold.}
\label{table: result_level1}
\resizebox{0.48\textwidth}{!}{
\begin{tabular} {l | c | c | c | c | c }
\toprule
\textbf{Method} & \textbf{CT} & \textbf{HC} & \textbf{PU} & \textbf{SP} & \textbf{Overall QAA} \\
\midrule
\multicolumn{6}{c}{\textit{Foundation Models}} \\
\midrule
Qwen3.7-Max~\cite{Qwen}  & 0.79 & 0.79 & 0.95 & 0.98 & 0.87 \\
DeepSeek-V4-Pro~\cite{DeepSeek} & 0.74 & 0.71 & 0.91 & 0.96 & 0.83 \\
Claude-Opus-4.8~\cite{Claude} & 0.76 & 0.73 & 0.87 & 0.92 & 0.82 \\
GPT-5.5~\cite{GPT} & 0.61 & 0.87 & 0.69 & 0.58 & 0.69 \\
\midrule
\multicolumn{6}{c}{\textit{Open-Source Multi-Agent Systems}} \\
\midrule
AutoIoT~\cite{AutoIoT} & 0.79 & 0.78 & 0.94 & 0.94 & 0.86 \\
IoTPilot~\cite{IoTPilot} & 0.75 & 0.76 & 0.95 & 0.95 & 0.85 \\
\midrule
\multicolumn{6}{c}{\textit{CLI Agent with MATLAB Agentic Toolkit}} \\
\midrule
Claude Code~\cite{ClaudeCode} + MATLAB ToolKit~\cite{MatlabToolkit} & 0.79 & 0.74 & \textbf{0.97} & 0.97 & 0.87 \\
Codex~\cite{Codex} + MATLAB ToolKit~\cite{MatlabToolkit} & 0.76 & 0.75 & 0.95 & 0.98 & 0.86 \\
\midrule
\rowcolor{gray!20}\textbf{\sysname (Ours)} & \textbf{0.97} & \textbf{0.97} & 0.96 & \textbf{0.99} & \textbf{0.97} \\
\bottomrule
\end{tabular}
}
\vspace{-5mm}
\end{table}

\subsection{Main Results}\label{section:5.3}
\textbf{Task Level 1.} As reported in Table~\ref{table: result_level1}, all methods answer foundational questions competently, with QAA from $0.69$ to $0.87$. \sysname attains the highest $0.97$, widening its margin most on the knowledge-intensive dimensions of communication theory and hardware constraints, where others stay roughly $0.6-0.9$. Thus, factual recall is not the bottleneck; the difficulty arises only when knowledge must become a physically faithful emission.

\textbf{Task Level 2.} Fig.~\ref{fig:mainresult_level2} exposes the multiplicative collapse of Section~\ref{section:3.1}: baselines degrade sharply along the \texttt{CPR$\rightarrow$HDR$\rightarrow$SIR} funnel. Equipping coding CLIs with the MATLAB Agentic Toolkit raises CPR via executable feedback, with Codex improving from $0.24$ to $0.41$, yet their SIR still collapses to at most $0.02$, since passing a compiler does not ensure a valid over-the-air emission. In contrast, \sysname sustains the funnel with only gentle decay, reaching CPR $0.88$, HDR $0.85$, and SIR $0.79$. Its SIR thus exceeds the best baseline by over $20\times$, with the gap largest at the strictest over-the-air level, where the collapse is otherwise total.

\textbf{Task Level 3.} Under ambiguous requirements and stringent physical constraints (Fig.~\ref{fig:mainresult_level3}), every baseline is pinned to the floor, with SIR at most $0.02$ across all sub-benchmarks. \sysname instead retains an overall SIR of $0.67$, degrading only mildly from $0.79$ at Level 2, while its CPR and HDR stay at $0.79$ and $0.77$. This graceful degradation holds across both the MATLAB- and UHD-based paradigms and across Wi-Fi and BLE. Grounding and staged retries lift each per-stage factor, so the end-to-end product no longer collapses as tasks harden.

\begin{figure}[t]
    \centering
    \includegraphics[width=0.99\linewidth]{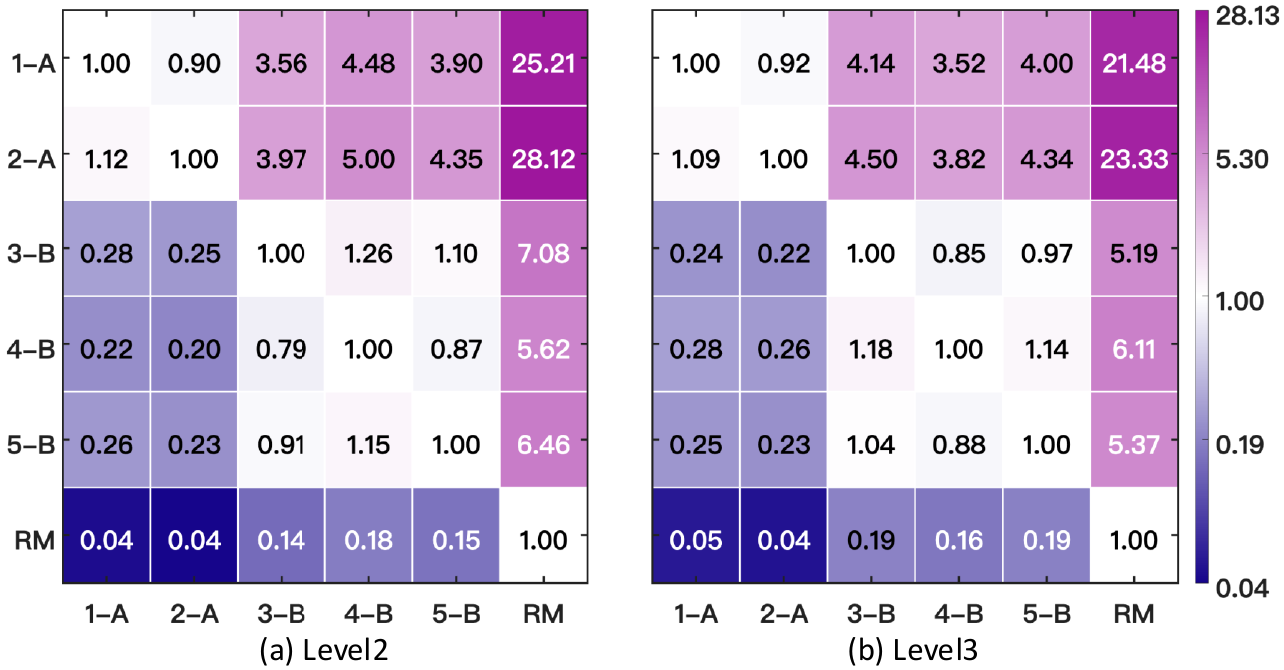}
    \caption{Human Experts vs. RadioMaster.}
    \label{fig:human_eval}
    \vspace{-5mm}
\end{figure}
\subsection{Human Experts vs. RadioMaster}\label{section:5.4}
We assess the practical value of \sysname by comparing its end-to-end configuration efficiency against human experts. We recruit five engineers with wireless and SDR backgrounds and randomly sample ten cases from each of \textbf{\textit{Task Level 2}} and \textbf{\textit{Task Level 3}}, spanning the MATLAB- and UHD-based paradigms. Two experts (Group A, 1-A and 2-A) work unaided, while three (Group B, 3-B, 4-B, and 5-B) are assisted by Codex. For each case, we log the wall-clock time from intent to a configuration whose waveform is captured and decoded by the target receiver, covering specification lookup, script authoring, debugging, and on-hardware deployment. \sysname (RM) performs the same tasks fully autonomously under the same criterion, and Fig.~\ref{fig:human_eval} reports the pairwise ratio of average configuration time.

\textbf{Results.} \sysname achieves an order-of-magnitude speedup on both levels. It configures a task up to $28\times$ faster than the unaided Group A, and, although Codex-assisted Group B is markedly faster than Group A, \sysname still outpaces it by up to $7\times$. This persistent gap shows that generic agentic coding tools alone cannot match the domain grounding, multi-agent collaboration, and closed-loop physical-layer verification of \sysname. The advantage holds from Level 2 to the harder Level 3, confirming the robustness of our pipeline. This study involved no ethical concerns.

\subsection{Ablation Studies}\label{section:5.5}
We ablate the three modules of \sysname to isolate how each counters the multiplicative collapse (Table~\ref{table: ablation}). The \emph{Direct Response} row is the monolithic single-agent baseline of Eq.~\eqref{eq:collapse}: its QAA stays high at $0.87$, yet SIR collapses to $0.01$, an empirical realization of $S_{\star}=\prod_{k}p_{k}\to0$.

Removing \textbf{RadioWiki} depresses every metric at once, lowering QAA from $0.97$ to $0.90$ and Level-2 SIR from $0.79$ to $0.45$, consistent with its role of lifting each factor from $p_{k}$ to $\tilde{p}_{k}$ via the recovery fraction $g_{k}$; once grounding is withdrawn, the upstream factors shrink and the whole product contracts. Removing \textbf{RadioAgent} keeps QAA intact at $0.97$ but sharply lowers CPR and HDR, with Level-2 HDR falling from $0.85$ to $0.59$, confirming that the detect-and-retry mechanism $\delta_{k},r_{k}$ of Theorem~\ref{thm:chain} acts on execution and deployment correctness rather than knowledge. Removing \textbf{RadioEmulator} leaves the upstream metrics almost unchanged while SIR alone collapses from $0.79$ to $0.36$, exactly the selective outcome Theorem~\ref{thm:gate} predicts, isolating the gate as the sole mechanism downstream of hardware deployment. Across the three ablations, each removed module reopens a distinct point of the multiplicative collapse so that only the full system sustains the SIR.

\begin{table}[t]
\centering
\caption{Ablation Study Results.}
\label{table: ablation}
\vspace{-3mm}
\resizebox{0.49\textwidth}{!}{
\begin{tabular}{ l | c | cc | cc | cc }
\toprule
\multirow{2}{*}{\textbf{Method}} & \multicolumn{1}{c}{\textbf{QAA}} & \multicolumn{2}{|c|}{\textbf{CPR}} & \multicolumn{2}{c|}{\textbf{HDR}} & \multicolumn{2}{c}{\textbf{SIR}} \\
& \multicolumn{1}{c}{Level 1} & \multicolumn{1}{|c}{Level 2} & \multicolumn{1}{c|}{Level 3} & \multicolumn{1}{c}{Level 2} & \multicolumn{1}{c|}{Level 3} & \multicolumn{1}{c}{Level 2} & \multicolumn{1}{c}{Level 3} \\
\midrule
Direct Response & 0.87 & 0.24 & 0.26 & 0.24 & 0.22 & 0.01 & 0.01 \\
\midrule
RadioMaster & 0.97 & 0.88 & 0.79 & 0.85 & 0.77 & 0.79 & 0.67 \\
~~w/o RadioWiki & 0.90 & 0.76 & 0.64 & 0.67 & 0.53 & 0.45 & 0.31 \\
~~w/o RadioAgent & 0.97 & 0.70 & 0.63 & 0.59 & 0.46 & 0.44 & 0.29 \\
~~w/o RadioEmulator & 0.97 & 0.80 & 0.78 & 0.78 & 0.76 & 0.36 & 0.23 \\
\bottomrule
\end{tabular}
}
\vspace{-3mm}
\end{table}

%% file: 6_Conclusion.tex
\section{Conclusion}
In this paper, we present \sysname, a fully autonomous multi-agent framework that closes the intent-to-air loop, translating high-level user intent into verified over-the-air emissions. Since the task forms a tightly chained process that collapses when any stage fails, \sysname counters this fragility with three synergistic modules: RadioWiki for knowledge grounding, RadioAgent for pipeline decomposition and local recovery, and RadioEmulator for closed-loop verification gating. We evaluate \sysname on RadioBench, the first benchmark for autonomous radio signal generation. Extensive real-world experiments show that \sysname consistently outperforms state-of-the-art baselines in configuration viability and signal fidelity, with emissions decoded by commercial receivers over the air. \sysname thus takes a concrete step toward bridging high-level reasoning and cyber-physical wireless systems, paving the way for AI-driven rapid prototyping in next-generation communications.

%% file: reference.bib
@inproceedings{AutoIoT,
  title={Autoiot: Llm-driven automated natural language programming for aiot applications},
  author={Shen, Leming and Yang, Qiang and Zheng, Yuanqing and Li, Mo},
  booktitle={Proceedings of the 31st Annual International Conference on Mobile Computing and Networking},
  pages={468--482},
  year={2025}
}

@inproceedings{IoTPilot,
  title={Programming Embedded IoT Applications in Natural Language with IoTPilot},
  author={Gong, Kaijie and Dong, Wei and Wang, Hao and Peng, Yingqi and Gao, Yi},
  booktitle={Proceedings of the 23rd Annual International Conference on Mobile Systems, Applications and Services},
  pages={70--82},
  year={2025}
}

@misc{Qwen,
title = {{Qwen3-Max}},
note = {\url{https://qwen.ai/blog?id=qwen3.7}}
}

@misc{MatlabToolkit,
title = {{Matlab Agentic Toolkit}},
note = {\url{https://www.mathworks.com/products/matlab-agentic-toolkit.html}}
}

@article{DeepSeek,
  title={Deepseek-v4: Towards highly efficient million-token context intelligence},
  author={Xu, Anyi and Lin, Bangcai and Xue, Bing and Wang, Bingxuan and Xu, Bingzheng and Wu, Bochao and Zhang, Bowei and Lin, Chaofan and Dong, Chen and Ling, Chenchen and others},
  journal={arXiv preprint arXiv:2606.19348},
  year={2026}
}

@misc{Claude,
title = {{Claude-Opus-4.6}},
note = {\url{https://platform.claude.com/docs/en/about-claude/models/whats-new-claude-4-8}}
}

@misc{GPT,
title = {{GPT-5.5}},
note = {\url{https://openai.com/zh-Hans-CN/index/introducing-gpt-5-5/}}
}

@misc{ClaudeCode,
title = {{Claude Code}},
note = {\url{https://code.claude.com/docs/en/overview}}
}

@misc{Codex,
title = {{Codex}},
note = {\url{https://openai.com/zh-Hans-CN/codex/}}
}

@misc{UHD,
title = {{Ettus USRP Hardware Driver}},
note = {\url{https://www.ettus.com/sdr-software/uhd-usrp-hardware-driver/
}}
}

@inproceedings {tinysdr,
author = {Mehrdad Hessar and Ali Najafi and Vikram Iyer and Shyamnath Gollakota},
title = {TinySDR: Low-Power SDR Platform for Over-the-Air Programmable IoT Testbeds },
booktitle = {17th USENIX Symposium on Networked Systems Design and Implementation (NSDI 20)},
year = {2020},
isbn = {978-1-939133-13-7},
address = {Santa Clara, CA},
pages = {1031--1046},
url = {https://www.usenix.org/conference/nsdi20/presentation/hessar},
publisher = {USENIX Association},
month = feb
}

@article{sdr1,
  title={Combining software-defined radio learning modules and neural networks for teaching communication systems courses},
  author={Camunas-Mesa, Luis A and de la Rosa, Jos'e M},
  journal={Information},
  volume={14},
  number={11},
  pages={599},
  year={2023},
  publisher={MDPI}
}

@ARTICLE{sdr2,
  author={Sklivanitis, George and Gannon, Adam and Batalama, Stella N. and Pados, Dimitris A.},
  journal={IEEE Communications Magazine}, 
  title={Addressing next-generation wireless challenges with commercial software-defined radio platforms}, 
  year={2016},
  volume={54},
  number={1},
  pages={59-67},
  doi={10.1109/MCOM.2016.7378427}}

@article{zou2026rf,
  title={RF-GPT: Teaching AI to See the Wireless World},
  author={Zou, Hang and Tian, Yu and Wang, Bohao and Bariah, Lina and Lasaulce, Samson and Huang, Chongwen and Debbah, M{\'e}rouane},
  journal={arXiv preprint arXiv:2602.14833},
  year={2026}
}

@article{chen2025radiollm,
  title={Radiollm: Introducing large language model into cognitive radio via hybrid prompt and token reprogrammings},
  author={Chen, Shuai and Zu, Yong and Feng, Zhixi and Yang, Shuyuan and Li, Mengchang},
  journal={arXiv preprint arXiv:2501.17888},
  year={2025}
}

@article{wang2025bridging,
  title={Bridging physical and digital worlds: embodied large AI for future wireless systems},
  author={Wang, Xinquan and Zhu, Fenghao and Yang, Zhaohui and Huang, Chongwen and Chen, Xiaoming and Zhang, Zhaoyang and Muhaidat, Sami and Debbah, M{\'e}rouane},
  journal={arXiv preprint arXiv:2506.24009},
  year={2025}
}

@article{cheng2025embodied,
  title={Embodied Intelligent Wireless (EIW): Synesthesia of Machines Empowered Wireless Communications},
  author={Cheng, Xiang and Wen, Weibo and Zhang, Haotian and Liu, Boxun and Yang, Zonghui and Zhang, Jianan and Cai, Xuesong},
  journal={arXiv preprint arXiv:2511.22845},
  year={2025}
}

@article{zhao2026agentic,
  title={An agentic system for rare disease diagnosis with traceable reasoning},
  author={Zhao, Weike and Wu, Chaoyi and Fan, Yanjie and Qiu, Pengcheng and Zhang, Xiaoman and Sun, Yuze and Zhou, Xiao and Zhang, Shuju and Peng, Yu and Wang, Yanfeng and others},
  journal={Nature},
  pages={1--10},
  year={2026},
  publisher={Nature Publishing Group UK London}
}

@inproceedings{liu2025tasksense,
  title={TaskSense: A Translation-like Approach for Tasking Heterogeneous Sensor Systems with LLMs},
  author={Liu, Kaiwei and Yang, Bufang and Xu, Lilin and Guo, Yunqi and Xing, Guoliang and Shuai, Xian and Ren, Xiaozhe and Jiang, Xin and Yan, Zhenyu},
  booktitle={Proceedings of the 23rd ACM Conference on Embedded Networked Sensor Systems},
  pages={213--225},
  year={2025}
}

@inproceedings{xu2024penetrative,
  title={Penetrative AI: Making LLMs Comprehend the Physical World},
  author={Xu, Huatao and Han, Liying and Yang, Qirui and Li, Mo and Srivastava, Mani},
  booktitle={Findings of the Association for Computational Linguistics: ACL 2024},
  pages={7324--7341},
  year={2024}
}

@inproceedings{zhao2025flexifly,
  title={FlexiFly: interfacing the physical world with foundation models empowered by reconfigurable drone systems},
  author={Zhao, Minghui and Xia, Junxi and Hou, Kaiyuan and Liu, Yanchen and Xia, Stephen and Jiang, Xiaofan},
  booktitle={Proceedings of the 23rd ACM Conference on Embedded Networked Sensor Systems},
  pages={463--476},
  year={2025}
}

@article{moler2020history,
  title={A history of MATLAB},
  author={Moler, Cleve and Little, Jack},
  journal={Proceedings of the ACM on Programming Languages},
  volume={4},
  number={HOPL},
  pages={1--67},
  year={2020},
  publisher={ACM New York, NY, USA}
}

@article{blossom2004gnu,
  title={GNU radio: tools for exploring the radio frequency spectrum},
  author={Blossom, Eric},
  journal={Linux journal},
  volume={2004},
  number={122},
  pages={4},
  year={2004},
  publisher={Belltown Media Houston, TX}
}

@inproceedings{huang2025chat3gpp,
  title={Chat3gpp: An open-source retrieval-augmented generation framework for 3gpp documents},
  author={Huang, Long and Zhao, Ming and Xiao, Limin and Zhang, Xiujun and Hu, Jungang},
  booktitle={2025 IEEE International Conference on Communications Workshops (ICC Workshops)},
  pages={492--497},
  year={2025},
  organization={IEEE}
}

@article{maatouk2024tele,
  title={Tele-llms: A series of specialized large language models for telecommunications},
  author={Maatouk, Ali and Ampudia, Kenny Chirino and Ying, Rex and Tassiulas, Leandros},
  journal={arXiv preprint arXiv:2409.05314},
  year={2024}
}

@article{thakur2024verigen,
  title={Verigen: A large language model for verilog code generation},
  author={Thakur, Shailja and Ahmad, Baleegh and Pearce, Hammond and Tan, Benjamin and Dolan-Gavitt, Brendan and Karri, Ramesh and Garg, Siddharth},
  journal={ACM Transactions on Design Automation of Electronic Systems},
  volume={29},
  number={3},
  pages={1--31},
  year={2024},
  publisher={ACM New York, NY}
}

@inproceedings{wang2024nn,
  title={NN-Defined Modulator: Reconfigurable and Portable Software Modulator on IoT Gateways},
  author={Wang, Jiazhao and Jiang, Wenchao and Liu, Ruofeng and Hu, Bin and Gao, Demin and Wang, Shuai},
  booktitle={21st USENIX Symposium on Networked Systems Design and Implementation (NSDI 24)},
  pages={775--789},
  year={2024}
}

@inproceedings{guo2024large,
  title={Large Language Model based Multi-Agents: A Survey of Progress and Challenges.},
  author={Guo, T and Chen, X and Wang, Y and Chang, R and Pei, S and Chawla, NV and Wiest, O and Zhang, X},
  booktitle={33rd International Joint Conference on Artificial Intelligence (IJCAI 2024)},
  year={2024},
  organization={IJCAI; Cornell arxiv}
}

@article{li2024survey,
  title={A survey on LLM-based multi-agent systems: workflow, infrastructure, and challenges},
  author={Li, Xinyi and Wang, Sai and Zeng, Siqi and Wu, Yu and Yang, Yi},
  journal={Vicinagearth},
  volume={1},
  number={1},
  pages={9},
  year={2024},
  publisher={Springer}
}

@inproceedings{yao2022react,
  title={React: Synergizing reasoning and acting in language models},
  author={Yao, Shunyu and Zhao, Jeffrey and Yu, Dian and Du, Nan and Shafran, Izhak and Narasimhan, Karthik R and Cao, Yuan},
  booktitle={The eleventh international conference on learning representations},
  year={2022}
}

@article{xu2023rewoo,
  title={Rewoo: Decoupling reasoning from observations for efficient augmented language models},
  author={Xu, Binfeng and Peng, Zhiyuan and Lei, Bowen and Mukherjee, Subhabrata and Liu, Yuchen and Xu, Dongkuan},
  journal={arXiv preprint arXiv:2305.18323},
  year={2023}
}

@article{guo2026embodied,
  title={Embodied llm agents learn to cooperate in organized teams},
  author={Guo, Xudong and Huang, Kaixuan and Liu, Jiale and Fan, Wenhui and V{\'e}lez, Natalia and Wu, Qingyun and Wang, Huazheng and Griffiths, Thomas L and Wang, Mengdi},
  journal={IEEE Transactions on Computational Social Systems},
  year={2026},
  publisher={IEEE}
}

@inproceedings{tan2020multi,
  title={Multi-agent embodied question answering in interactive environments},
  author={Tan, Sinan and Xiang, Weilai and Liu, Huaping and Guo, Di and Sun, Fuchun},
  booktitle={European Conference on Computer Vision},
  pages={663--678},
  year={2020},
  organization={Springer}
}

@article{wu2025generative,
  title={Generative multi-agent collaboration in embodied ai: A systematic review},
  author={Wu, Di and Wei, Xian and Chen, Guang and Shen, Hao and Wang, Xiangfeng and Li, Wenhao and Jin, Bo},
  journal={arXiv preprint arXiv:2502.11518},
  year={2025}
}

@article{ghafarollahi2025sciagents,
  title={SciAgents: automating scientific discovery through bioinspired multi-agent intelligent graph reasoning},
  author={Ghafarollahi, Alireza and Buehler, Markus J},
  journal={Advanced Materials},
  volume={37},
  number={22},
  pages={2413523},
  year={2025},
  publisher={Wiley Online Library}
}

@inproceedings{fan2025ai,
  title={Ai hospital: Benchmarking large language models in a multi-agent medical interaction simulator},
  author={Fan, Zhihao and Wei, Lai and Tang, Jialong and Chen, Wei and Siyuan, Wang and Wei, Zhongyu and Huang, Fei},
  booktitle={Proceedings of the 31st International Conference on Computational Linguistics},
  pages={10183--10213},
  year={2025}
}

@inproceedings{peng2025tree,
  title={Tree-of-Reasoning: Towards Complex Medical Diagnosis via Multi-Agent Reasoning with Evidence Tree},
  author={Peng, Qi and Cui, Jialin and Xie, Jiayuan and Cai, Yi and Li, Qing},
  booktitle={Proceedings of the 33rd ACM International Conference on Multimedia},
  pages={1744--1753},
  year={2025}
}

@misc{tong2026wirelessagentautomatedagenticworkflow,
      title={WirelessAgent++: Automated Agentic Workflow Design and Benchmarking for Wireless Networks}, 
      author={Jingwen Tong and Zijian Li and Fang Liu and Wei Guo and Jun Zhang},
      year={2026},
      eprint={2603.00501},
      archivePrefix={arXiv},
      primaryClass={cs.NI},
      url={https://arxiv.org/abs/2603.00501}, 
}

@inproceedings{ma2026autorf,
  title={AutoRF: Towards an Agentic Framework for Automated RF Hardware Design},
  author={Ma, Ruichun and Qiu, Lili and Hu, Wenjun and Wang, Jiazhao and Song, Yiwen and Pan, Hao},
  booktitle={Proceedings of the 24th Annual International Conference on Mobile Systems, Applications and Services},
  pages={279--293},
  year={2026}
}

@inproceedings{luo2026emerging,
  title={IoTGen: Towards LLM-driven IoT Hardware Generation},
  author={Luo, Qinpei and Ma, Ruichun and Zhang, Xinyu and Qiu, Lili},
  booktitle={Proceedings of the 24th Annual International Conference on Mobile Systems, Applications and Services},
  pages={792--807},
  year={2026}
}

@inproceedings{yang2026autoembed,
  title={AutoEmbed: LLM-driven Automated Software Development for Generic Embedded IoT Systems},
  author={Yang, Huanqi and Li, Mingzhe and Han, Mingda and Li, Zhenjiang and Xu, Weitao},
  booktitle={Proceedings of the 2026 ACM/IEEE International Conference on Embedded Artificial Intelligence and Sensing Systems},
  pages={362--376},
  year={2026}
}

@inproceedings{wang2026ppai,
  title={PPAI: Enabling Personalized LLM Agent Interoperability for Collaborative Edge Intelligence},
  author={Wang, Zile and Liu, Qianli and Guo, Kaibin and Wang, Haodong and Lin, Jian and Hong, Zicong and Guo, Song},
  booktitle={IEEE INFOCOM 2026-IEEE Conference on Computer Communications},
  pages={1--10},
  year={2026},
  organization={IEEE}
}

@inproceedings{kwon2026open,
  title={Open RAN Conflict Agents: Detecting and Mitigating xApp Conflicts with Generative Agents},
  author={Kwon, Dae Cheol and Zhang, Xinyu},
  booktitle={IEEE INFOCOM 2026-IEEE Conference on Computer Communications},
  pages={1--10},
  year={2026},
  organization={IEEE}
}

@misc{lei2026enablingagileambientiot,
      title={Enabling Agile Ambient IoT Networking via a Parameterized Hybrid Radio}, 
      author={Jiazhen Lei and Fengyuan Zhu and Tianze Cao and Yuxin Sha and Linling Zhong and Wenhui Li and Bingbing Wang and Zeming Yang and Jinyang Sun and Yibin Deng and Xiaohua Tian},
      year={2026},
      eprint={2605.18314},
      archivePrefix={arXiv},
      primaryClass={cs.NI},
      url={https://arxiv.org/abs/2605.18314}, 
}
